\documentclass{article}
\usepackage[cmex10]{amsmath}
\usepackage{amssymb}
\usepackage{graphicx}
\usepackage{hyperref}
\usepackage{fullpage}

\newtheorem{theorem}{Theorem}[section]
\newtheorem{lemma}[theorem]{Lemma}
\newtheorem{observation}{Observation}
\newtheorem{problem}{Problem}

\newcommand{\sq}{\hbox{\rlap{$\sqcap$}$\sqcup$}}
\newcommand{\qed}{\hspace*{\fill}\sq}
\newenvironment{proof}{\noindent {\bf Proof.}\ }{\qed\par\vskip 4mm\par}

\graphicspath{{./figures/}}

\begin{document}
\title{Shortest Paths of Mutually Visible Robots}
\author{Rusul J. Alsaedi, Joachim Gudmundsson, André van Renssen}
\date{}
\maketitle

\begin{abstract}
Given a set of $n$ point robots inside a simple polygon $P$, the task is to move the robots from their starting positions to their target positions along their shortest paths, while the mutual visibility of these robots is preserved. Previous work only considered two robots. In this paper, we present an $O(mn)$ time algorithm, where $m$ is the complexity of the polygon, when all the starting positions lie on a line segment $S$, all the target positions lie on a line segment $T$, and $S$ and $T$ do not intersect. We also argue that there is no polynomial-time algorithm, whose running time depends only on $n$ and $m$, that uses a single strategy for the case where $S$ and $T$ intersect.
\end{abstract}

\section{Introduction}
\label{section:Introduction}

Mobile robots are widely used for surveillance, cleaning, inspection and transportation tasks, with applications in fields such as household maintenance, space exploration, delivery of goods and services, and wastewater treatment~\cite{ben2018robots,nehmzow2001mobile}.

The problem of navigating robots in a polygonal domain has been studied extensively. In this paper, we assume that robots are independent point robots and execute the same algorithm. The robots cannot communicate with each other, but they have full visibility of their environment. We study the following problem in this model:   

\begin{problem} {\sc Shortest Path Mutual Visibility (SPMV)}
Given a simple polygon $P$ consisting of $m$ vertices, a set $\mathcal{R}$ of $n$ point robots, where each robot $r_i \in \mathcal{R}$ has a starting position $s_i$ and a target position $t_i$, the problem is to move all robots from their starting positions to their target positions along their shortest paths while keeping all robots mutually visible.   
\end{problem}

In Figure~\ref{fig:two_cases} we show two examples for the restricted case where the starting points lie on a segment $S$ and the target points lie on a segment $T$. 

\begin{figure}[ht!]
 \centering
  \includegraphics[width=\textwidth]{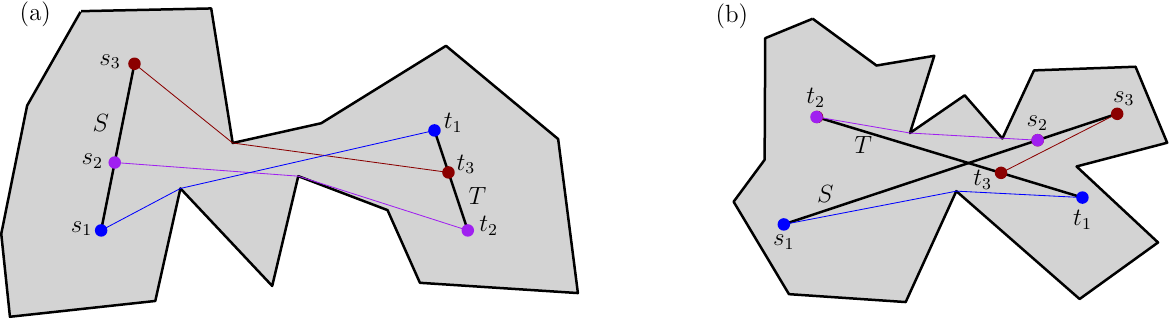}
  \caption{Two examples showing the problem instances considered in this paper. (a) The case when $S$ and $T$ do not intersect, and (b) when $S$ and $T$ intersect.}
  \label{fig:two_cases}
\end{figure}

Most of the existing work in the field focuses on the shortest path problem. Guibas {\it et al.}~\cite{guibas1987linear} presented algorithms to solve the shortest path problem between two points inside a simple polygon and showed the relationship between visibility and the shortest path problem. See also~\cite{ahn2016linear,aichholzer2014geodesic,aronov2015time,asano2013memory,bae2013geodesic,bae2019computing,bae2015computing,barba2015space} for more recent work on shortest paths in simple polygons and polygonal domains.

The mutual visibility problem for robots in a polygonal domain has been widely studied; starting from any initial configuration, the aim is to move the robots into positions where every robot can see all other robots. The problem was solved under both obstructed visibility (i.e., three collinear robots are not mutually visible to each other) and unobstructed visibility (i.e., three collinear robots are mutually visible to each other)~\cite{https://doi.org/10.48550/arxiv.2206.14423,di2017mutual,Luna2014,poudel2019sublinear,Sharma2015,sharma2018make,Vaidyanathan2015}. 

However, these papers do not consider mutual visibility while the robots are moving. The most closely related work is by Fenwick {\it et al.}~\cite{fenwick2005optimal,fenwick2007mutually}. They considered the {\sc SPMV} problem for the special case where $n=2$ (i.e., there are two moving robots the need to remain mutually visible while moving through the domain), whereas we consider the problem for $n$, i.e., an arbitrary number of robots. 

Like Fenwick {\it et al.}, we consider a discretization of the robot movements. We aim to compute a schedule, which specifies for each robot at what moment it moves where (in a straight line from its current location) and at what speed. We consider only constant movement speed in this paper. When needed, we will refer to a \emph{step}, meaning a maximal period of time during which no changes to the robots speed or direction occur. In other words, if robot $r_1$ is moving with speed 2 towards the right and robots $r_2$ and $r_3$ are currently not moving, then this step ends when either robot $r_1$ stops, changes direction, or changes speed, or when robot $r_2$ or $r_3$ start moving. 

The results presented in this paper are twofold. In Section~\ref{section:Non-Crossing Start and Target Lines} we give an $O(nm)$ time algorithm for the {\sc SPMV} problem when the start and target positions lie on two non-intersecting line segments $S$ and $T$, respectively. In Section~\ref{section:Crossing Start and Target Lines} we consider the case where $S$ and $T$ intersect. Intuitively speaking, we argue that any efficient algorithm must move the robots along a rotating line, however, we also show that there are instances where moving the robots along a rotating line does not give a running time that can be bounded using only $n$ and $m$. Hence, this indicates that any polynomial-time algorithm for the {\sc SPMV} problem cannot rely on a single strategy.

\section{Non-Crossing Start and Target Line Segments}
\label{section:Non-Crossing Start and Target Lines}
In this section, we present an algorithm that solves the \textsc{SPMV} problem when the start and target segments, $S$ and $T$, do not intersect. 

Let $\pi_i$ be the shortest path within $P$ from the start point $s_i$ on $S$ of robot $r_i$ to its target point $t_i$ on $T$, for $1 \leq i \leq n$. Let $Q$ be the minimal, possibly degenerate polygon within $P$ that includes the start line $S$, the target line $T$, and the set of $\pi_i$'s, $1 \leq i \leq n$, as shown in Figure~\ref{fig:Defining_Q}.

If all paths $\pi_i$, $1 \leq i \leq n$, share a single polygon vertex, we refer to this vertex as their intersection point, as shown in Figure~\ref{fig:intersection_point}(a). If they share multiple vertices, then we say that the first of these vertices along the paths is the intersection point, see Figure~\ref{fig:intersection_point}(b). 

\begin{figure}[ht!]
 \centering
  \includegraphics[width=\textwidth]{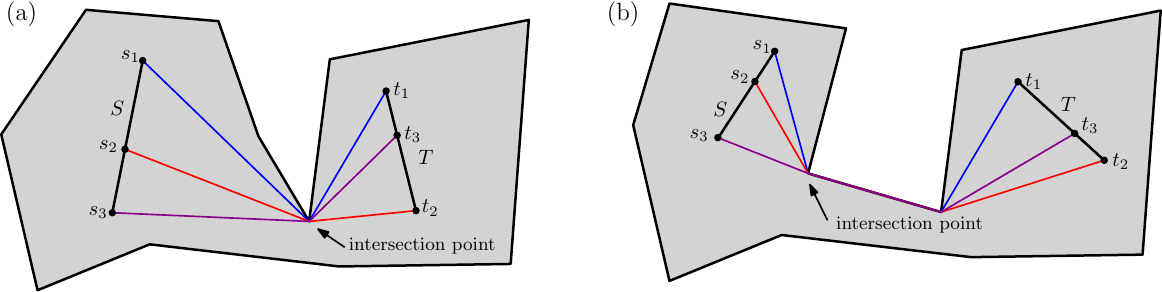}
  \caption{Illustrating the definition of an intersection point.}
  \label{fig:intersection_point}
\end{figure}

The following observation is immediate.
\begin{observation}\label{obs:Q}
If the shortest paths have no intersection point then $Q$ consists of $S$, $T$ and two concave paths connecting the end points of $S$ and $T$, see Figure~\ref{fig:Defining_Q}(a).
If the paths $\pi_i$, $1\leq i \leq n$, have an intersection point then $Q$ is a degenerate polygon consisting of two funnels; one containing $S$ and one containing $T$, and the two funnels are connected by either a joint intersection point or a single path, as shown in Figure~\ref{fig:Defining_Q}(b).
\end{observation}

\begin{figure}[ht!]
 \centering
  \includegraphics[width=\textwidth]{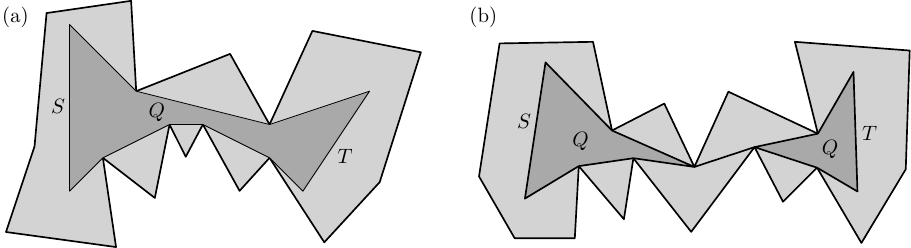}
  \caption{Illustrating Observation~\ref{obs:Q}. (a) $Q$ consists of $S$, $T$ and two concave chains connecting $S$ and $T$. (b) $Q$ is a degenerate polygon consisting of two funnels; one containing $S$ and one containing $T$, and the two funnels are connected by either a joint intersection point or a single polygonal path.}
  \label{fig:Defining_Q}
\end{figure}

To simplify the description of the algorithm, rotate $Q$ so that $S$ is entirely to the left of $T$ and neither $S$ nor $T$ are horizontal. Furthermore, assume without loss of generality that $s_1, \ldots, s_n$ are ordered from top-to-bottom along $S$. 

Let $U(Q)=\langle s_1=u_1, u_2, \ldots , u_k \rangle$ denote the top boundary of $Q$ connecting $s_1$ with the top-most endpoint of $T$ and let $V(Q) =\langle s_n=v_1, v_2, \ldots , v_{\ell}\rangle$ denote the bottom boundary of $Q$ connecting $s_n$ with the bottom-most endpoint of $T$. Note that if $Q$ is a degenerate polygon with an intersection point then at least one vertex of $Q$ will be in both $U(Q)$ and $V(Q)$.

Our algorithm starts by computing a triangulation $\cal{T}$ of $Q$ in linear time~\cite{c91}. Note that since $Q$ might contain an intersection point, a triangle in $\cal{T}$ might degenerate to a single segment.

All robots are initially positioned on $S$. Let $\tau$ be the unique triangle in $\cal{T}$ with the segment $S=(s_1,s_n)$ as one of its sides and $u_2$ and/or $v_2$ as its third vertex, according to Observation~\ref{obs:Q}. Assume without loss of generality that $u_2$ is the third vertex of $\tau$. 

With the above notation, we next describe the algorithm.  

Consider the initial configuration when all the robots are positioned on $S$. If $u_2=v_2$ then move every robot $r_i$ along $\pi_i$ to $u_2$.
Otherwise, move every robot $r_i$ along $\pi_i$ from the segment $(s_1,s_n)$ to the segment $(s_n,u_2)$. More formally, we move robot $r_i$ to the intersection of $(s_n,u_2)$ and $\pi_i$. We note that since $u_2$ lies on $U(Q)$ and $s_n$ lies on $V(Q)$, every $\pi_i$ intersects $(s_n,u_2)$. 

In a generic step of the algorithm, the robots are positioned on a (possibly degenerate) segment $(u_a,v_b)$ of $\cal{T}$. The algorithm considers four cases:   
\begin{description}
    \item[(a)] If $u_a=v_b$ and $u_{a+1}=v_{b+1}$ then every robot $r_i$ moves from $u_a$ to $u_{a+1}$.
    \item[(b)] If $u_a=v_b$ and $u_{a+1}\neq v_{b+1}$ then every robot $r_i$ moves along $\pi_i$ from $u_a$ to the segment $(u_{a+1}, v_{b+1})$.
    \item[(c)] If $u_a\neq v_b$ and $u_{a+1} = v_{b+1}$ then every robot $r_i$ moves along $\pi_i$ from $(u_a,v_b)$ to $u_{a+1}$.
    \item[(d)] Otherwise, if $u_a\neq v_b$ and $u_{a+1} \neq  v_{b+1}$ then let $\tau$ be the unique triangle in $\cal{T}$ with $(u_a, v_b)$ as one of its sides and either $u_{a+1}$ or $v_{b+1}$ as its third vertex. Assume without loss of generality that $u_{a+1}$ is the third vertex of $\tau$. Move every robot $r_i$ along $\pi_i$ from segment $(u_a,v_b)$ to segment $(u_{a+1}, v_b)$. 
\end{description} 
This iterative process continues until all robots have reached $T$.

The algorithm runs in $O(mn)$ time since the triangulation can be computed in linear time, and the shortest paths computation requires $O(mn)$ time.

It remains only to prove that the robots are mutually visible to each other during their movement from $S$ to $T$.

\begin{lemma}    
The algorithm described above guarantees that every robot $r_i$ moves from $s_i$ to $t_i$ along $\pi_i$ while being visible to every other robot.
\end{lemma}
\begin{proof}
Within each (possibly degenerate) triangle $\tau$ of $\cal T$, the robots move along their shortest paths from one segment (or point) to another segment (or point) of $\tau$. Since a triangle is convex, all robots are mutually visible within $\tau$. Therefore, what is left to prove is that all the shortest paths visit the same set of triangles in the exact same order. 
From Observation~\ref{obs:Q}, it follows that every non-degenerate triangle in $\cal T$ has at least one vertex from $V(Q)\setminus U(Q)$ and one vertex from $U(Q)\setminus V(Q)$. That implies that (1) any path from a point on $S$ to a point on $T$ must intersect every triangle in $\cal T$, and (2) every robot $r_i$ following $\pi_i$ from $s_i$ to $t_i$ must visit the triangles in the same order. Consequently, all the robots must visit all the triangles in the same order, and as a result, every $r_i$ is visible to every other robot during the movement along their shortest paths.
\end{proof}

To summarize this section, we obtain the following theorem.
\begin{theorem} \label{thm:non-intersecting}
    In the case when the starting positions and the target positions lie on two non-intersecting segments, the {\textsc{SPMV}} problem can be solved in $O(nm)$ time.
\end{theorem}


\section{Crossing Start and Target Line Segments}
\label{section:Crossing Start and Target Lines} 
At first sight one might think that the case when $S$ and $T$ intersect is not much harder than the non-intersecting case. However, we will argue that it is. Consider the instance shown in Figure~\ref{fig:2.1}(a). It is clear (Lemma~\ref{lemma:22}) that any algorithm that does not move the robots along a rotating line around the intersection point $q$ between $S$ and $T$ cannot terminate in polynomial time with respect to the input size (assuming $v_2$ and $v_5$ are very close to $q$). However, any algorithm that keeps the robots on a line pivoting around $q$ can also get stuck (Lemma~\ref{lemma:11}). This implies that any algorithm for the case when $S$ and $T$ intersect cannot use a single strategy algorithm, i.e., it will have to analyse the input instance and depending on the instance select an appropriate approach. 

We will need the following notations before formalising the above observations.

\begin{figure}[ht!]
 \centering
  \includegraphics[width=\textwidth]{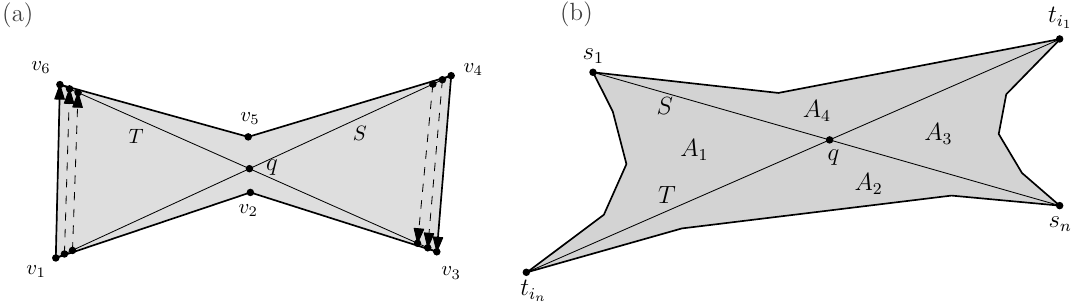}
  \caption{(a) An instance showing that the number of steps of any algorithm that does not move all the robots at the same time cannot be bounded by $n$ and $m$.  (b) An example illustrating the concave paths together with the regions $A_1, A_2, A_3$, and $A_4$.}
  \label{fig:2.1}
\end{figure}

Let $q$ be the intersection of the start and target line segments $S$ and~$T$, as illustrated in Figure~\ref{fig:2.1}(b). 
Define $Q$ as in the previous section. Furthermore, assume that $s_1, \ldots, s_n$ are ordered from top-to-bottom along $S$, and let $t_{i_1}, \ldots, t_{i_n}$ denote the target points along $T$ ordered from top-to-bottom. 

Let $A_1$ be the simple polygon bounded by the two segments $(s_1,q)$ and $(t_{i_n},q)$, and the concave chain corresponding to the shortest path in $Q$ from $s_1$ to $t_{i_n}$. Symmetrically, we define $A_2$, $A_3$, and $A_4$, see Figure~\ref{fig:2.1}(b).

Note that a robot $r_i$ moving from $s_i \in S$ to $t_i \in T$ will stay within one of the regions $A_1, A_2, A_3$, or $A_4$ during their movement. Hence, we can partition the robots into four sets $\mathcal{R}_1$, $\mathcal{R}_2$, $\mathcal{R}_3$, and $\mathcal{R}_4$ such that $\mathcal{R}_j$, $1\leq j \leq 4$, contains the robots in $\mathcal{R}$ moving through $A_j$. If a robot $r_i$ only moves along the boundary of two regions (not in their interior), then $r_i$ can be assigned to either of the two sets.

Next, we argue that any algorithm for the case when $S$ and $T$ intersect that does not move all robots at the same time requires an unbounded number of steps.

\begin{lemma}
\label{lemma:22}
There exist instances where any algorithm that does not move all robots of $\mathcal{R}_1$ and $\mathcal{R}_3$ at the same time requires a number of steps that is unbounded in $n$ and $m$.
\end{lemma}
\begin{proof}
Consider the instance shown in Figure~\ref{fig:2.1}(a). It consists of six vertices $v_1, \ldots , v_6$, where $v_1$, $v_3$, $v_4$ and $v_6$ are at the endpoints of $S$ and $T$ while $v_2$ and $v_5$ are arbitrarily close to $q$.

Let all robots in $\mathcal{R}_1$ start arbitrarily close to $v_1$ and move to a position arbitrarily close to $v_6$, and let all robots in $\mathcal{R}_3$ start arbitrarily close to $v_4$ and move arbitrarily close to $v_3$.

We observe that for the robots in $\mathcal{R}_1$ and $\mathcal{R}_3$ to be visible to each other, they all need to be contained in a narrow strip defined by two almost\footnote{As $v_2$ and $v_5$ approach $q$, these lines become more and more parallel.} parallel lines through $v_2$ and $v_5$.

Hence, in each step where not all robots are moved, the robots can move at most twice the width of this narrow strip. 
Since this width is determined by the distance of $v_2$ and $v_5$ to $q$, this can be arbitrarily small, leading to a number of steps not related to $n$ and $m$.
\end{proof}

Next, we prove that any algorithm that moves the robots along a rotating line around $q$ can also get stuck.

\begin{lemma} \label{lemma:11}
If $S$ and $T$ intersect, then any algorithm that keeps the robots in $\mathcal{R}_1$ and $\mathcal{R}_3$ on a line through $q$ or that keeps $\mathcal{R}_2$ and $\mathcal{R}_4$ on a line through $q$ can get stuck and be unable to proceed.
\end{lemma}
\begin{proof}
We first show that any algorithm that keeps the robots in $\mathcal{R}_1$ and $\mathcal{R}_3$ on a line can get stuck, before extending this to also include $\mathcal{R}_2$ and $\mathcal{R}_4$.

Let $a$ and $b$ denote the top and bottom endpoints of $S$, respectively, and let $c$ and $d$ denote the top and bottom endpoints of $T$, respectively.
We have four robots: $r_1$ starting at $a$ and moving to $d$, $r_2$ starting at $a$ and moving to $c$, $r_3$ starting at $b$ and moving to $c$, and $r_4$ starting at $b$ and moving to $d$.

\begin{figure}[ht!]
 \centering
  \includegraphics[width=\textwidth]{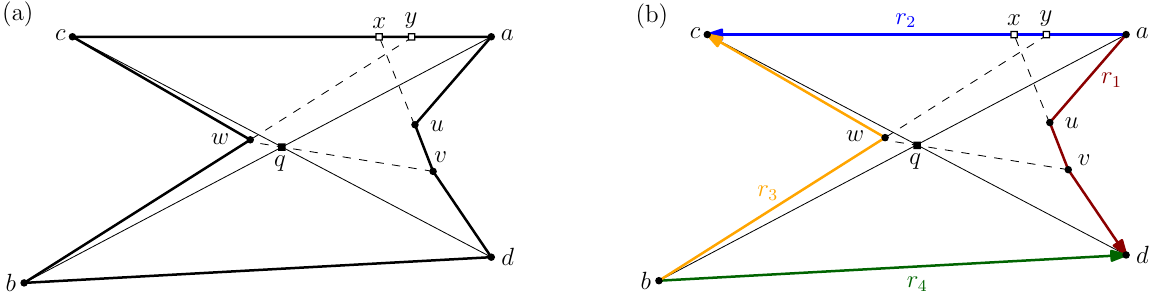}
  \caption{(a) An input instance showing that any algorithm that keeps the robots on a line through $q$ can reach a situation where it is unable to proceed. (b) Illustrating the shortest paths for the four robots.}
  \label{fig:2.2}
\end{figure}

Consider the polygon in Figure~\ref{fig:2.2}(a) where $a, b, c$, and $d$ are vertices of $P$,  $(a,c)$ and $(b,d)$ are edges of the polygon, $a$ and $d$ are connected by a concave chain $\langle a,u,v,d \rangle$, and $b$ and $c$ are connected by a concave chain $\langle b,w,c \rangle$.

In order to ensure that $r_1$ and $r_2$ remain visible, $r_1$ cannot move past $v$ before $r_2$ is on or past the projection $x$ of $\overline{uv}$ onto $(a,c)$. This also means that since $r_1$ and $r_3$ lie on a line through $q$, $r_3$ cannot reach $w$ before $r_1$ passes $v$. However, to maintain visibility between $r_2$ and $r_3$, $r_2$ needs to be on or before the projection $y$ of $\overline{bw}$ onto $(a,c)$.

We now observe that since $y$ occurs before $x$ along $(a,c)$, to maintain visibility, $r_2$ needs to be both on or before $y$ and on or after $x$, meaning that there is no position along $(a,c)$ that maintains visibility of $r_2$ with both $r_1$ and $r_3$ if we require $r_1$ and $r_3$ to be on a line through $q$.

However, in the above example, keeping $r_2$ and $r_4$ on a line through $q$ would still work, so we now extend the above construction to create a similar situation for that line while maintaining the one for $r_1$ and $r_3$. The construction is shown in Figure~\ref{fig:2.3}(a).

\begin{figure}[ht!]
 \centering
  \includegraphics[width=\textwidth]{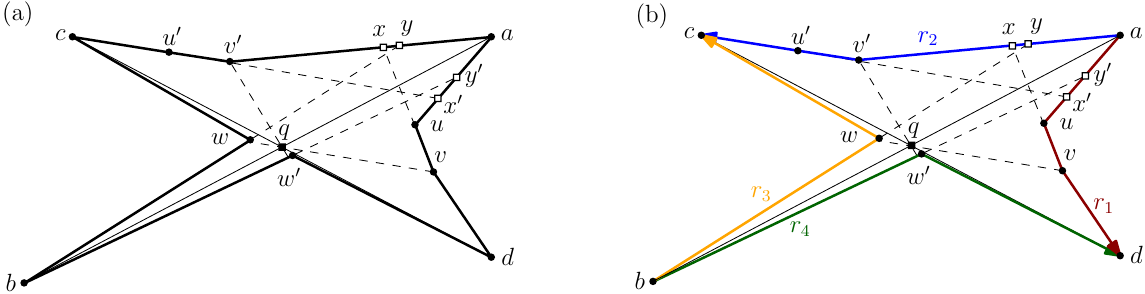}
  \caption{(a) An input instance showing that any algorithm for the crossing case that keeps the robots on a line through $q$ can reach a situation where it is unable to proceed. (b) Illustrating the shortest paths for the four robots.}
  \label{fig:2.3}
\end{figure}

Compared to the construction in Figure~\ref{fig:2.2}, we replace the edge $(a,c)$ with the concave chain $\langle a, v', u', c \rangle$, and the edge $(b,d)$ with the concave chain $\langle b, w', d \rangle$. Using an analogous argument to the previous one, $y$ still comes before $x$ along the path from $a$ to $c$, and using $u'$, $v'$ and $w'$, we can construct $x'$ and $y'$ analogous to their counterparts to ensure that $y'$ comes before $x'$ along the path from $a$ to $d$. Thus, neither the line connecting $r_1$ and $r_3$ through $q$, nor the line connecting $r_2$ and $r_4$ through $q$ can move from $v$ or $v'$, respectively.

Since we assumed nothing of the algorithm other than that it keeps either $\mathcal{R}_1$ and $\mathcal{R}_3$ or $\mathcal{R}_2$ and $\mathcal{R}_4$ on a line through $q$, we can conclude that any such algorithm will reach a situation in which it is unable to proceed.
\end{proof}


The above results indicate that any polynomial-time algorithm for the {\sc SPMV} problem cannot rely on a single strategy. Instead, it must take the input instance into account when choosing which strategy to choose. 

The {\sc SPMV} problem requires that all robots move along their shortest path and that the robots are visible to each other at all time.  We note that if we remove either condition then polynomial-time algorithms are trivial.

\section{Concluding Remarks}
\label{Concluding Remarks}
In this paper, we considered the {\sc SPMV} problem for $n$ point robots. We gave an efficient $O(nm)$ time algorithm for the case where the start and target line segments do not intersect. We also argued that any polynomial-time algorithm for the case where these line segments do intersect has to take the input polygon into consideration to determine the strategy to move the robots. In particular, we showed that no single strategy will give a polynomial-time algorithm for all simple polygons. 

There are still many open problems related to the {\sc SPMV} problem. Most importantly, is there a polynomial-time algorithm for the general case, or even for the case when $S$ and $T$ intersect? If there is not, can one prove any lower bounds? 

Additionally, we could relax the requirements on the paths the robots follow or their visibility. For example, if we allow approximate shortest paths, what bounds can be derived? Or if we require a robot to only be visible to a fraction of the other robots or even allow robots to be non-visible for a certain duration or travel distance? 

We could also allow more flexibility in the types of movement we allow for the robots. For example, we could allow the speed to vary during a step by, for example, defining it as a low-degree polynomial. We could also allow the robots to move along trajectories defined by low-degree polynomials or trigonometric functions. These types of relaxations give rise to many different related problems that could be considered as future work. 

\bibliographystyle{plain}
\bibliography{Bibliography}
\end{document}